\documentclass[runningheads, envcountsame]{llncs}%
\usepackage{etex}
\usepackage{verbatim}
\usepackage{amsfonts}
\usepackage{amsmath}
\usepackage{tikz}
\usepackage{array}
\usepackage{xcolor}
\usepackage{url}
\usepackage{lineno}
  \newcommand{\defproblem}[3]{
  \vspace{2mm}
\noindent\fbox{
  \begin{minipage}{0.96\textwidth}
  #1\\
  {\bf{Input:}} #2  \\
  {\bf{Output:}} #3
  \end{minipage}
  }
  \vspace{2mm}
}

\usepackage[algo2e, noend, noline, boxed]{algorithm2e}
  \SetKwComment{tcm}{\{}{\}}

\newcommand{\F}{\mathcal{F}}
\newcommand{\MAW}{\textsc{MAW-SequenceComparison}}
\newcommand{\CMAW}{\textsc{MAW-CircularSequenceComparison}}
\def\dd{\mathinner{.\,.}}
\newcommand{\cO}{\mathcal{O}}
\newcommand{\SA}{\textsf{SA}}
\newcommand{\iSA}{\textsf{iSA}}

\newcommand{\LCP}{\textsf{LCP}}
\newcommand{\LCE}{\textsf{LCE}}
\newcommand{\lcp}{\textsf{lcp}}

\newcommand{\mina}{\mathcal{M}}
\newcommand{\LW}{\textsf{LW}}

\begin{document}
%
%
%

\sloppy

\title{Linear-Time Sequence Comparison Using Minimal Absent Words \& Applications}

\author{Maxime Crochemore\inst{1}
\and Gabriele Fici\inst{2}
\and Robert Merca\c{s}\inst{1,3}
\and Solon P.\ Pissis\inst{1}
}

\institute{ 
$\!^1\ $Department of Informatics, King's College London, UK\\
$\!^2\ $Dipartimento di Matematica e Informatica, Universit{\`a} di Palermo, Italy\\
$\!^3\ $Department of Computer Science, Kiel University, Germany\\
\email{maxime.crochemore@kcl.ac.uk, gabriele.fici@unipa.it, rgm@informatik.uni-kiel.de, solon.pissis@kcl.ac.uk} \\ 
}

\titlerunning{Linear-Time Sequence Comparison Using Minimal Absent Words}
\authorrunning{M. Crochemore, G. Fici, R. Merca\c{s}, S. Pissis}

\maketitle
\vspace{-.7em}
\begin{abstract}
Sequence comparison is a prerequisite to virtually all comparative genomic analyses.
It is often realized by sequence alignment techniques, which are computationally expensive.
This has led to increased research into alignment-free techniques, which are based on measures referring to the composition
of sequences in terms of their constituent patterns. These measures, such as $q$-gram distance, are usually computed in time linear 
with respect to the length of the sequences. 
In this article, we focus on the complementary idea: how two sequences can be efficiently compared based
on information that does not occur in the sequences.  
A word is an {\em absent word} of some sequence if it does not occur in the sequence.
An absent word is {\em minimal} if all its proper factors occur in the sequence.
Here we present the first linear-time and linear-space algorithm 
to compare two sequences by considering {\em all} their minimal absent words.
In the process, we present results of combinatorial interest, and also extend the proposed techniques to compare circular sequences.

\end{abstract}

\section{Introduction}

Sequence comparison is an important step in many basic tasks in bioinformatics, from phylogenies reconstruction to genomes assembly. 
It is often realized by sequence alignment techniques, which are computationally expensive, requiring quadratic time in the length of the sequences. 
This has led to increased research into \textit{alignment-free} techniques. 
Hence standard notions for sequence comparison are gradually being complemented and in some cases replaced by alternative ones~\cite{Domazet-Loso:2009:EEP:1671627.1671629}. 
One such notion is based on comparing the words that are absent in each sequence~\cite{nullrly}. 
A word is an \textit{absent word} (or a forbidden word) of some sequence if it does not occur in the sequence. Absent words represent a type of \textit{negative information}: information about what does not occur in the sequence. 

Given a sequence of length $n$, the number of absent words of length at most $n$ is exponential in $n$. However, the number of certain classes of absent words is only linear in $n$.
This is the case for \textit{minimal absent words}, that is, absent words in the sequence whose all proper factors occur in the sequence~\cite{BeMiReSc00}. 
An upper bound on the number of minimal absent words is known to be $\cO(\sigma n)$~\cite{Crochemore98automataand,Mignosi02}, where $\sigma$ is the size of the alphabet $\Sigma$. 
Hence it may be possible to compare sequences in time proportional to their lengths, for a fixed-sized alphabet, instead of proportional to the product of their lengths. 
In what follows, we consider sequences on a {\em fixed-sized alphabet} since the most commonly studied alphabet is $\Sigma=\{\texttt{A,C,G,T}\}$.

An $\cO(n)$-time and $\cO(n)$-space algorithm for computing all minimal absent words on a fixed-sized alphabet based on the construction of suffix automata 
was presented in~\cite{Crochemore98automataand}.  
The computation of minimal absent words based on the construction of suffix arrays was considered in~\cite{Pinho2009}; although this algorithm has a linear-time performance in practice, 
the worst-case time complexity is $\cO(n^2)$. New $\cO(n)$-time and $\cO(n)$-space suffix-array-based algorithms were presented in~\cite{DBLP:conf/isit/FukaeOM12,MAW,PPAM2015} to bridge this unpleasant gap. 
An implementation of the algorithm presented in~\cite{MAW} is currently, and to the best of our knowledge, the fastest available for the computation of minimal absent words.
A more space-efficient solution to compute all minimal absent words in time $\cO(n)$ was also presented in~\cite{Belazzougui2013}.

In this article, we consider the problem of comparing two sequences $x$ and $y$ of 
respective lengths $m$ and $n$, using their sets of minimal absent words.
In~\cite{Chairungsee2012109}, Chairungsee and Crochemore 
introduced a measure of similarity between two sequences based on the notion of minimal absent words. 
They made use of a length-weighted index to provide a measure of similarity between two sequences, using sample sets of their minimal absent words, by considering the length of
each member in the symmetric difference of these sample sets. This measure can be trivially computed in time and space $\cO(m + n)$ provided that these sample sets contain
minimal absent words of some bounded length $\ell$. For unbounded length, the same measure can be trivially computed in time $\cO(m^2 + n^2)$: for a given sequence, 
the cumulative length of all its minimal absent words can grow {\em quadratically} with respect to the length of the sequence.

The same problem can be considered for two {\em circular} sequences. The measure of similarity of  Chairungsee and Crochemore can be used in this setting provided that one extends the 
definition of minimal absent words to circular sequences. In Section~\ref{sec:circ_seq_comp}, we give a definition of minimal absent words for a circular sequence from the Formal Language 
Theory point of view. We believe that this definition may also be of interest from the point of view of Symbolic Dynamics, which is the original context in 
which minimal absent words have been introduced~\cite{BeMiReSc00}.

\medskip
\noindent \textbf{Our Contribution.} Here we make the following threefold contribution:%
\vspace{-0.2cm}
\begin{description}
 \item[a)] We present an $\cO(m + n)$-time and $\cO(m + n)$-space algorithm to compute the similarity measure introduced by Chairungsee and Crochemore by considering {\em all} minimal absent words of two  sequences $x$ and $y$ of lengths $m$ and $n$, respectively; thereby
 showing that it is indeed possible to compare two sequences in time proportional to their lengths (Section~\ref{sec:seq_comp}).
 \item[b)]  We show how this algorithm can be applied to compute this similarity measure for two circular sequences $x$ and $y$ of lengths $m$ and $n$, respectively, in the same time and space complexity as a result of the extension of the definition of minimal absent words to circular sequences (Section~\ref{sec:circ_seq_comp}).
 \item[c)]  We provide an open-source code implementation of our algorithms and investigate potential applications of our theoretical findings (Section~\ref{sec:imp_app}).
\end{description}

\section{Definitions and Notation}\label{sec:prem}

  We begin with basic definitions and notation.
  Let $y=y[0]y[1]\dd y[n-1]$ be a \textit{word} of \textit{length} $n=|y|$
over a finite ordered \textit{alphabet} $\Sigma$ of size 
$\sigma = |\Sigma|=\cO(1)$.
For two positions $i$ and $j$ on $y$, we denote by $y[i\dd j]=y[i]\dd y[j]$ the \textit{factor} 
(sometimes called \textit{substring}) of $y$ that 
starts at position $i$ and ends at position $j$ (it is empty if $j < i$), and by $\varepsilon$ 
the \textit{empty word}, word of length 0. 
  We recall that a prefix of $y$ is a factor that starts at position 0 
($y[0\dd j]$) and a suffix is a factor that ends at position $n-1$ 
($y[i\dd n-1]$), and that a factor of $y$ is a \textit{proper} factor if 
it is not $y$ itself. The set of all the factors of the word $y$ is denoted by $\F_y$.

  Let $x$ be a word of length $0<m\leq n$. 
  We say that there exists an \textit{occurrence} of $x$ in $y$, or, more 
simply, that $x$ \textit{occurs in} $y$, when $x$ is a factor of $y$.
  Every occurrence of $x$ can be characterised by a starting position in $y$. 
  Thus we say that $x$ occurs at the \textit{starting position} $i$ in $y$ 
when $x=y[i \dd i + m - 1]$.
  Opposingly, we say that the word $x$ is an \textit{absent word} of
$y$ if it does not occur in $y$.
  The absent word $x$ of $y$ is \textit{minimal} if and only if all its proper factors 
occur in $y$. The set of all minimal absent words for a word $y$ is denoted by $\mina_y$. 
For example, if $y=abaab$, then $\mina_y=\{aaa, aaba, bab, bb\}$. In general, if we suppose that 
all the letters of the alphabet appear in $y$ of length $n$, the length of a minimal absent word of $y$ lies between $2$ and $n+1$. 
It is equal to $n+1$ if and only if $y$ is the catenation of $n$ copies of the same letter. 
So, if $y$ contains occurrences of at least two different letters, the length of a minimal absent word for $y$ is bounded from above by $n$.

A {\em language} over the alphabet $\Sigma$ is a set of finite words over $\Sigma$. A language is {\em regular} if it is recognized by a finite state automaton. A language is {\em factorial} if it contains all the factors of its words. A language is {\em antifactorial} if no word in the language is a proper factor of another word in the language. Given a word $x$, the language \emph{generated} by $x$ is the language $x^*=\{x^k\mid k\geq 0\}=\{\varepsilon, x, xx, xxx, \ldots\}$. The \emph{factorial closure} of a language $L$ is the language $\F_L=\{\F_y\mid y\in L\}$.  Given a factorial language $L$, one can define the (antifactorial) language of minimal absent words for $L$ as $\mina_L=\{aub\mid aub\notin L, au,ub\in L\}$. 
Notice that $\mina_L$ is not the same language as the union of $\mina_x$ for $x\in L$. 

We denote by \SA{} the {\em suffix array} of $y$ of length $n$, that is, an integer array of size $n$
storing the starting positions of all (lexicographically) sorted suffixes of $y$, i.e.~for all 
$1 \leq  r < n$ we have $y[\SA{}[r-1] \dd n-1] < y[\SA{}[r] \dd n - 1]$~\cite{SA}.
  Let \lcp{}$(r, s)$ denote the length of the longest common prefix between
$y[\SA{}[r] \dd n - 1]$ and $y[\SA{}[s] \dd n - 1]$ 
for all positions $r$, $s$ on $y$, and $0$ otherwise.
  We denote by \LCP{} the {\em longest common prefix} array of $y$ defined by 
\LCP{}$[r]=\lcp{}(r-1, r)$ for all $1 \leq r < n$, and 
\LCP{}$[0] = 0$. The inverse \iSA{} of the array \SA{} is defined by 
$\iSA{}[\SA{}[r]] = r$, for all $0 \leq r < n$. It is known that
  \SA{}~\cite{Nong:2009:LSA:1545013.1545570}, \iSA{}, and 
\LCP{}~\cite{indLCP} of a word of length $n$  can be computed in time and space $\cO(n)$.

In what follows, as already proposed in~\cite{MAW}, for every word $y$, the set of minimal words 
associated with $y$, denoted by $\mina_y$, is represented as a set of tuples $\langle a, i,j \rangle$, 
where the corresponding minimal absent word $x$ of $y$ is defined by
$x[0]=a$, $a \in \Sigma$, and $x[1 \dd m-1] = y[i \dd j]$, where $j-i+1=m \geq 2$.
It is known that if $|y|=n$ and $|\Sigma|=\sigma$, then $|\mina_y| \leq \sigma n$~\cite{Mignosi02}.

In~\cite{Chairungsee2012109}, Chairungsee and Crochemore introduced a measure of similarity between two words $x$ and $y$ based on the notion of minimal absent words. 
Let $\mina_x^\ell$ (resp.~$\mina_y^\ell$) denote the set of minimal absent words of length at most $\ell$ of $x$ (resp.~$y$).
The authors made use of a length-weighted index to provide a measure of the similarity between $x$ and $y$, using their sample sets $\mina_x^\ell$ and $\mina_y^\ell$, by considering the length of
each member in the symmetric difference $(\mina_x^\ell \bigtriangleup \mina_y^\ell)$ of the sample sets. For sample sets $\mina_x^\ell$ and $\mina_y^\ell$, they defined this index to be
$$\LW(\mina_x^\ell,\mina_y^\ell) = \sum_{w \in \mina_x^\ell \bigtriangleup \mina_y^\ell} \frac{1}{|w|^2}.$$
This work considers the following generalized version of the same problem.

\defproblem{\MAW}{a word $x$ of length $m$ and a word $y$ of length $n$}{$\LW(\mina_x,\mina_y)$, where $\mina_x$ and $\mina_y$ 
denote the sets of minimal absent words of $x$ and $y$, respectively.}

We also consider the aforementioned problem for two circular words. A circular word of length $m$ can be viewed as a traditional linear word which has the left- and right-most letters 
wrapped around and stuck together in some way. Under this notion, the same circular word can be seen as $m$ different 
linear words, which would all be considered equivalent. More formally, given a word $x$ of length $m$, we denote 
by $x^{\langle i \rangle}=x[i \dd m-1]x[0 \dd i-1]$, $0 \leq i < m$, the $i$-th \textit{rotation} of $x$, where $x^{\langle 0 \rangle}=x$. Given two words $x$ and $y$, we define $x\sim y$ if and only if there exist $i$, $0 \leq i < |x|$, such that $y=x^{\langle i \rangle}$. A \emph{circular word} $\tilde{x}$ is a conjugacy class of the equivalence relation $\sim$. Given a circular word $\tilde{x}$, any (linear) word $x$ in the equivalence class $\tilde{x}$ is called a \emph{linearization} of the circular word $\tilde{x}$. Conversely, given a linear word $x$, we say that $\tilde{x}$ is a \emph{circularization} of $x$ if and only if $x$ is a linearization of $\tilde{x}$.
The set $\F_{\tilde{x}}$ of factors of the circular word $\tilde{x}$ is equal to the set $\F_{xx}\cap \Sigma^{\leq |x|}$ of factors of  $xx$ whose length is at most $|x|$, where $x$ is any linearization of $\tilde{x}$. 

Note that if $x^{\langle i \rangle}$ and $x^{\langle j \rangle}$ are two rotations of the same word, then the factorial languages  $\F_{(x^{\langle i \rangle})^*}$ and $\F_{(x^{\langle j \rangle})^*}$ coincide, so one can unambiguously define the (infinite) language $\F_{\tilde{x}^*}$ as the language $\F_{x^*}$, where $x$ is any linearization of $\tilde{x}$.

In Section~\ref{sec:circ_seq_comp}, we give the definition of the set $\mina_{\tilde{x}}$ of minimal absent words for a circular word $\tilde{x}$. 
We will prove that the following problem can be solved with the same time and space complexity as its counterpart in the linear case.

\defproblem{\CMAW}
{a word $x$ of length $m$ and a word $y$ of length $n$}{$\LW(\mina_{\tilde{x}},\mina_{\tilde{y}})$, 
where $\mina_{\tilde{x}}$ and $\mina_{\tilde{y}}$ 
denote the sets of minimal absent words of the circularizations $\tilde{x}$ of $x$ and $\tilde{y}$ of $y$, respectively.}

\section{Sequence Comparison}\label{sec:seq_comp}

The goal of this section is to provide the first linear-time and linear-space algorithm for computing the similarity measure (see Section~\ref{sec:prem}) between two words defined over a fixed-sized alphabet. 
To this end, we consider two words $x$ and $y$ of lengths $m$ and $n$, respectively, and their associated sets of minimal absent words, $\mina_x$ and $\mina_y$, respectively. 
Next, we give a linear-time and linear-space solution for the {\MAW} problem. 
It is known from~\cite{Crochemore98automataand} and~\cite{MAW} that we can compute the sets $\mina_x$ and $\mina_y$ in linear time and space with respect to the two lengths $m$ and $n$, respectively.
The idea of our strategy consists of a merge sort on the sets $\mina_x$ and $\mina_y$, after they have been ordered with the help of suffix arrays.

To this end, we construct the suffix array associated to the word $w=xy$, together with the implicit $\LCP$ array corresponding to it. 
All of these structures can be constructed in time and space $\cO(m+n)$, as mentioned earlier. Furthermore, we can preprocess the array \textsf{LCP} for range minimum queries, which we denote by $\textsf{RMQ}_\textsf{LCP}$~\cite{Fischer11}.
With the preprocessing complete, the longest common prefix $\LCE$ of two suffixes of $w$ starting at positions $p$ and $q$ can be computed in
constant time~\cite{LCE}, using the formula
$\LCE(w,p,q)=\textsf{LCP}[\textsf{RMQ}_{\textsf{LCP}}(\textsf{iSA}[p]+1,\textsf{iSA}[q])].$

Using these data structures, it is straightforward to sort the tuples in the sets $\mina_x$ and $\mina_y$ lexicographically. That is, two tuples $x_1,x_2\in \mina_x$, are ordered according to the letter following their longest common prefix, or when it is not the case, with the one being the prefix, coming first. To do this, we simply go once through the suffix array associated to $w$ and assign to each tuple in $\mina_x$, respectively $\mina_y$, the rank of the suffix starting at the position indicated by its second component, in the suffix array. 
Since sorting an array of $n$ distinct integers, such that each is in $[0,n-1]$, can be done in time $\cO(n)$ (using bucket sort, for example), we can sort now each of the sets of minimal absent words, taking into consideration the letter on the first position and these ranks. Thus, from now on, we assume that $\mina_x=(x_0, x_1,\ldots, x_k)$ where $x_i$ is  lexicographically smaller than $x_{i+1}$, for $0\leq i <k\leq \sigma m$, and $\mina_y=(y_0, y_1,\ldots, y_\ell)$, where $y_j$ is  lexicographically smaller than $y_{j+1}$, for $0\leq j <\ell\leq \sigma n$.

Provided these tools, we now proceed to do the merge. Thus, considering that we are analysing the $(i+1)$th tuple in $\mina_x$ and the $(j+1)$th tuple in $\mina_y$, we note that the two are equal if and only if $x_i[0]=y_j[0]$ and 
$$\LCE(w,x_i[1], |x|+y_j[1])\geq \ell, \mbox{ where } \ell=x_i[2]-x_i[1]=y_j[2]-y_j[1].$$ 
In other words, the two minimal absent words are equal if and only if their first letters coincide, they have equal length $\ell+1$, and the longest common prefix of the suffixes of $w$ starting at the positions indicated by 
the second components of the tuples has length at least $\ell$.

Such a strategy will empower us with the means for constructing a new set $\mina_{xy}=\mina_x\cup\mina_y$. At each step, when analysing tuples $x_i$ and $y_j$ we proceed as following:
$$
\mina_{xy} = \left\{
                \begin{array}{l c l}
                  \mina_{xy}\cup \{x_i\}, & \mbox{and increment } i, &\qquad\mbox{if } x_i < y_j;\\
                  \mina_{xy}\cup \{y_j\}, & \mbox{and increment } j, &\qquad\mbox{if } x_i > y_j;\\
                  \mina_{xy}\cup \{x_i=y_j\}, &   \mbox{and increment both } i \mbox{ and } j, &\qquad \mbox{if } x_i = y_j.
                \end{array}
              \right.
$$
Observe that the last condition is saying that basically each common tuple is added only once to their union.

Furthermore, simultaneously with this construction we can also calculate the similarity between the words, given by $\LW(\mina_x,\mina_y)$, which is initially set to $0$. 
Thus, at each step, when comparing the tuples $x_i$ and $y_j$, we update 
$$
\LW(\mina_x,\mina_y) = \left\{
                \begin{array}{l c l}
                  \LW(\mina_x,\mina_y) + \frac{1}{|x_i|^2}, & \mbox{and increment } i, &\mbox{if } x_i < y_j;\\
                  \LW(\mina_x,\mina_y) + \frac{1}{|y_j|^2}, & \mbox{and increment } j, &\mbox{if } x_i > y_j;\\
                  \LW(\mina_x,\mina_y), &  \mbox{and increment both } i \mbox{ and } j, &\mbox{if } x_i = y_j.
                \end{array}
              \right.
$$
We impose the increment of both $i$ and $j$ in the case of equality as in this case we only look at the symmetric difference between the sets of minimal absent words.

As all these operations take constant time, once per each tuple in $\mina_x$ and $\mina_y$, it is easily concluded that the whole operation takes in the case of a fixed-sized alphabet time and space $\cO(m+n)$. 
Thus, we can compute the symmetric difference between the {\em complete} sets of minimal absent words, as opposed to~\cite{Chairungsee2012109}, of two words defined over a fixed-sized alphabet, in linear time and space with respect to the lengths of the two words. We obtain the following result.

\begin{theorem}
\label{the:maw}
Problem \MAW{} can be solved in time and space $\cO(m+n)$.
\end{theorem}

\section{Circular Sequence Comparison}\label{sec:circ_seq_comp}

Next, we discuss two possible definitions for the minimal absent words of a circular word, and highlight the differences between them. 

We start by recalling some basic facts about minimal absent words. For further details and references the reader is recommended~\cite{fici}. Every factorial language $L$ is uniquely determined by its (antifactorial) language of minimal absent words $\mina_L$, through the equation $L=\Sigma^*\setminus \Sigma^*\mina_L\Sigma^*$. The converse is also true, since by the definition of a minimal absent word we have $\mina_L=\Sigma L\cap L\Sigma \cap (\Sigma^*\setminus L)$. The previous equations define a bijection between factorial and antifactorial languages. Moreover, this bijection preserves regularity. In the case of a single (linear) word $x$, the set of minimal absent words for $x$ is indeed the antifactorial language $\mina_{\F_{x}}$. Furthermore, we can retrieve $x$ from its set of minimal absent words in linear time and space \cite{Crochemore98automataand}. 

Recall that given a circular word $\tilde{x}$, the set $\F_{\tilde{x}}$ of factors of $\tilde{x}$ is equal to the set $\F_{xx}\cap \Sigma^{\leq |x|}$ of factors of $xx$ 
whose lengths are at most $|x|$, where $x$ is any linearization of $\tilde{x}$. 
Since a circular word $\tilde{x}$ is a conjugacy class containing all the rotations of a linear word $x$, the language $\F_{\tilde{x}}$ can be seen as the factorial closure of the set $\{x^{\langle i\rangle}\mid i=0,\ldots,|x|-1\}$. This leads to the first definition of the set of minimal absent words for $\tilde{x}$, that is the set $\mina_{\F_{\tilde{x}}}=\{aub\mid a,b\in \Sigma, aub\notin \F_{\tilde{x}}, au,ub\in \F_{\tilde{x}}\}$. For instance, if $x=abaab$, we have $$\mina_{\F_{\tilde{x}}}=\{aaa, aabaa, aababa, abaaba, ababaa, baabab, babaab, babab, bb\}.$$ 

The advantage of this definition is that we can retrieve uniquely $\tilde{x}$ from $\mina_{\F_{\tilde{x}}}$. However, the total size of $\mina_{\F_{\tilde{x}}}$ (that is, the sum of the lengths of its elements) can be very large, as the following lemma suggests.

\begin{lemma}
Let $\tilde{x}$ be a circular word of length $m>0$. The set $\mina_{\F_{\tilde{x}}}$ contains precisely $\ell$ words of maximal length $m+1$, where $\ell$ is the number of distinct rotations of any linearization $x$ of $\tilde{x}$, that is, the cardinality of $\{x^{\langle i\rangle}\mid i=0,\ldots,|x|-1\}$. 
\end{lemma}

\begin{proof}
Let $x=x[0]x[1] \dd x[m-1]$ be a linearization of $\tilde{x}$. The word obtained by appending to $x$ its first letter, $x[0]x[1] \dd x[m-1]x[0]$, belongs to $\mina_{\F_{\tilde{x}}}$, since it has length $m+1$, hence it cannot belong to $\F_{\tilde{x}}$, but its maximal proper prefix $x=x^{\langle 0\rangle}$ and its maximal proper suffix $x^{\langle 1\rangle}=x[1] \dd x[m-1]x[0]$ belong to $\F_{\tilde{x}}$. 

The same argument shows that for any rotation $x^{\langle i\rangle}=x[i]x[i+1] \dd x[m-1]x[0]\dd x[i-1]$ of $x$, the word $x[i]x[i+1] \dd x[m-1]x[0]\dd x[i-1]x[i]$, obtained by 
appending to $x^{\langle i\rangle}$ its first letter, belongs to $\mina_{\F_{\tilde{x}}}$.

Conversely, if a word of maximal length $m+1$ is in  $\mina_{\F_{\tilde{x}}}$, then its maximal proper prefix and its maximal proper suffix are words of length $m$ in $\F_{\tilde{x}}$, so they must be consecutive rotations of $x$.

Therefore, the number of words of maximal length $m+1$ in $\mina_{\F_{\tilde{x}}}$ equals the number of distinct rotations of $x$, hence the statement follows.
\qed
\end{proof}

This is in sharp contrast with the situation for linear words, where the set of minimal absent words can be represented on a trie having size linear in the length of the word. Indeed, the algorithm \textsc{MF-trie}, introduced in~\cite{Crochemore98automataand}, builds the tree-like deterministic automaton accepting the set of minimal absent words for a word $x$ taking as input the factor automaton of $x$, that is the minimal deterministic automaton recognizing the set of factors of $x$. The leaves of the trie correspond to the minimal absent words for $x$, while the internal states are those of the factor automaton. Since the factor automaton of a word $x$ has less than $2|x|$ states (for details, see~\cite{CHL07}), this provides a representation of the minimal absent words of a word of length $n$ in space $O(\sigma n)$.

This algorithmic drawback leads us to the second definition. This second definition of minimal absent words for circular strings has been already introduced in~\cite{DBLP:conf/isit/OtaM13,6979851}.
First, we give a combinatorial result which shows that when considering circular words it does not make sense to look at absent words obtained from more than two rotations.

\begin{lemma}\label{lem:general}
For any positive integer $k$ and any word $u$, the set $V=\{v \mid k|u|+1< |v| \leq (k+1)|u|\} \cap (\mina_{u^{k+1}}\setminus\mina_{u^k})$  is empty.
\end{lemma}
\begin{proof}
This obviously holds for all words $u$ of length 1.
Assume towards a contradiction that this is not the case in general. Hence, there must exist a word $v$ of length $m$ that fulfills the conditions in the lemma, thus $v\in V$ and $m>2$. Furthermore, since the length $m-1$ prefix and the length $m-1$ suffix of every minimal absent word occur in the main word at non-consecutive positions, there must exist positions $i<j\leq n=|u|$ such that 
\begin{eqnarray}\label{eq:1}
v[1\dd m-2]=u^{k+1}[i+1\dd i+m-2]=u^{k+1}[j+1\dd j+m-2].
\end{eqnarray}
Obviously, following Equation~(\ref{eq:1}), since $m-2\geq kn$, we have that $v[1 \dd m-2]$ is $(j-i)$-periodic. But, we know that $v[1\dd m-2]$ is also $n$-periodic. 
Thus, following a direct application of the periodicity lemma we have that $v[1\dd m-2]$ is $p=\gcd(j-i,n)$-periodic. 
But, in this case we have that $u$ is $p$-periodic, and, therefore, $u[i]=u[j]$, which leads to a contradiction with the fact that $v$ is a minimal absent word, whenever $i$ is defined. Thus, it must be the case that $i=-1$. Using the same strategy and looking at positions $u[i+m-2]$ and $u[j+m-2]$, we conclude that $j+m-2=(k+1)n$. Therefore, in this case, we have that $m=kn+1$, which is a contradiction with the fact that the word $v$ fulfills the conditions of the lemma.
This concludes the proof.
\qed\end{proof}

Observe now that the set $V$ consists in fact of all extra minimal absent words generated whenever we look at more than one rotation, that do not include the length arguments. That is, $V$ does not include the words bounding the maximum length that a word is allowed, nor the words created, or lost, during a further concatenation of an image of $u$. However, when considering an iterative concatenation of the word, these extra elements determined by the length constrain cancel each other.

As observed in Section~\ref{sec:prem}, two rotations of the same word $x$ generate two languages that have the same set of factors. So, we can unambiguously associate to a circular word $\tilde{x}$ the (infinite) factorial language $\F_{\tilde{x}^*}$.
It is therefore natural to define the set of minimal absent words for the circular word $\tilde{x}$ as the set $\mina_{\F_{\tilde{x}^*}}$. For instance, if $\tilde{x}=abaab$, then we have 
$$\mina_{\F_{\tilde{x}^*}}=\{aaa, aabaa, babab, bb\}.$$ 

This second definition is much more efficient in terms of space, as we show below. In particular, the length of the words in $\mina_{\F_{\tilde{x}^*}}$ is bounded from above by $|x|$, 
hence  $\mina_{\F_{\tilde{x}^*}}$ is a finite set.

Recall that a word $x$ is \emph{a power} of a word $y$ if there exists a positive integer $k>1$ such that $x$ is expressed as $k$ consecutive concatenations of $y$, denoted by $x=y^k$. 
Conversely, a word $x$ is {\em primitive} if $x=y^k$ implies $k=1$. Notice that a word is primitive if and only if any of its rotation is. We can therefore extend the definition of primitivity to circular words. The definition of $\mina_{\F_{\tilde{x}^*}}$ does not allow one to uniquely reconstruct $\tilde{x}$ from $\mina_{\F_{\tilde{x}^*}}$, unless $\tilde{x}$ is known to be primitive, since it is readily verified that $\F_{\tilde{x}^*}=\F_{\widetilde{xx}^*}$ and therefore also the minimal absent words of these two languages coincide. However, from the algorithmic point of view, 
this issue can be easily managed by storing the length $|x|$ of a linearization $x$ of $\tilde{x}$ together with the set $\mina_{\F_{\tilde{x}^*}}$.
Moreover, in most practical cases, for example when dealing with biological sequences, it is highly unlikely that the circular word considered is not primitive.

The difference between the two definitions above is presented in the next lemma.

\begin{lemma}\label{lem:twodef}
 $\mina_{\F_{\tilde{x}^*}}=\mina_{\F_{\tilde{x}}}\cap \Sigma^{\leq |x|}.$
\end{lemma}

\begin{proof}
 Clearly, $\F_{\tilde{x}^*}\cap \Sigma^{\leq |x|}=\F_{\tilde{x}}$. 
 The statement then follows from the definition of minimal absent words.\qed
\end{proof}

Based on the previous discussion, we set $\mina_{\tilde{x}}=\mina_{\F_{\tilde{x}^*}}$, while the following corollary comes straightforwardly as a consequence of Lemma~\ref{lem:general}.

\begin{corollary}
  \label{lem:circ}
  Let $\tilde{x}$ be a circular word. Then $\mina_{\tilde{x}}=\mina_{xx}^{|x|}$. 
\end{corollary}

Corollary~\ref{lem:circ} was first introduced as a definition for the set of minimal absent words of a circular word in~\cite{6979851}. 
Using the result of Corollary~\ref{lem:circ}, we can easily extend the algorithm described in the previous section to the case of circular words. 
That is, given two circular words $\tilde{x}$ of length $m$ and  $\tilde{y}$ of length $n$, we can compute in time and space $\cO(m+n)$ 
the quantity $\LW(\mina_{\tilde{x}},\mina_{\tilde{y}})$. We obtain the following result.

\begin{theorem}
\label{the:cmaw}
Problem \CMAW{} can be solved in time and space $\cO(m+n)$.
\end{theorem}

\section{Implementation and Applications}\label{sec:imp_app}

We implemented the presented algorithms as programme \textsf{scMAW} 
to perform pairwise sequence comparison for a set of sequences using minimal absent words.
\textsf{scMAW} uses programme \textsf{MAW}~\cite{MAW} for linear-time and linear-space computation of minimal absent words using suffix array. 
\textsf{scMAW} was implemented in the $\textsf{C}$ programming language and developed under GNU/Linux operating system. 
It takes, as input argument, a file in MultiFASTA format with the input sequences, and then any of the two methods, for {\em linear} or {\em circular} sequence comparison, 
can be applied. It then produces a file in PHYLIP format with the distance matrix as output.
Cell $[x,y]$ of the matrix stores $\LW(\mina_{x},\mina_{y})$ (or $\LW(\mina_{\tilde{x}},\mina_{\tilde{y}})$ for the circular case).
The implementation is distributed under the GNU General Public License (GPL), and it is available at \url{http://github.com/solonas13/maw}, 
which is set up for maintaining the source code and the man-page documentation. 
Notice that {\em all} input datasets and the produced outputs referred to in this section are publicly maintained at the same web-site.

An important feature of the proposed algorithms is that they require space linear in the length of the sequences 
(see Theorem~\ref{the:maw} and Theorem~\ref{the:cmaw}). Hence, we were also able to implement \textsf{scMAW}
using the Open Multi-Processing (OpenMP) PI for shared memory multiprocessing programming to distribute the workload 
across the available processing threads without a large memory footprint.

\noindent \textbf{Application.}
Recently, there has been a number of studies on the biological significance of absent words in various species~\cite{nullrly,minabpro,Silva02042015}.
In~\cite{minabpro}, the authors presented dendrograms from dinucleotide relative abundances in sets of minimal absent words for prokaryotes and eukaryotic genomes.
The analyses support the hypothesis that minimal absent words are inherited through a common ancestor, in addition to lineage-specific inheritance, 
only in vertebrates. Very recently, in~\cite{Silva02042015}, it was shown that there exist three minimal words in the Ebola virus genomes which are absent from human genome.
The authors suggest that the identification of such species-specific sequences may prove to be useful for the development of both diagnosis and therapeutics.

In this section, we show a potential application of our results for the construction of dendrograms for DNA sequences with circular structure.
Circular DNA sequences can be found in viruses, as plasmids in archaea and bacteria, and in the mitochondria and plastids of eukaryotic cells.
Circular sequence comparison thus finds applications in several contexts such as reconstructing phylogenies using viroids RNA~\cite{conf/gcb/MosigHS06} or Mitochondrial DNA (MtDNA)~\cite{mtDNA-phylo}. 
Conventional tools to align circular sequences could yield an incorrectly high genetic distance between closely-related species. Indeed, when sequencing
molecules, the position where a circular sequence starts can be totally arbitrary. Due to this {\em arbitrariness}, a suitable rotation of one sequence would give much better results 
for a pairwise alignment~\cite{SEA2015,WABI2015}. In what follows, we demonstrate the power of minimal absent words to pave a path to resolve 
this issue by applying Corollary~\ref{lem:circ} and Theorem~\ref{the:cmaw}.
Next we do not claim that a solid phylogenetic analysis is presented but rather an investigation for potential applications of our theoretical findings.

We performed the following experiment with synthetic data. 
First, we simulated a basic dataset of DNA sequences using INDELible~\cite{indelible}.
The number of taxa, denoted by $\alpha$, was set to $12$; 
the length of the sequence generated at the root of the tree, denoted by $\beta$, was set to 2500bp;
and the substitution rate, denoted by $\gamma$, was set to $0.05$. 
We also used the following parameters: a deletion rate, denoted by $\delta$, of $0.06$ \textit{relative} to substitution rate of $1$; 
and an insertion rate, denoted by $\epsilon$, of $0.04$ \textit{relative} to substitution rate of $1$. 
The parameters were chosen based on the genetic diversity standard measures observed for sets of MtDNA sequences from primates and mammals~\cite{SEA2015}. 
We generated another instance of the basic dataset, containing one {\em arbitrary} rotation of each of the $\alpha$ sequences from the basic dataset. 
We then used this randomized dataset as input to \textsf{scMAW} by considering $\LW(\mina_{\tilde{x}},\mina_{\tilde{y}})$ as the distance metric. 
The output of \textsf{scMAW} was passed as input to \textsf{NINJA}~\cite{ninja}, an efficient implementation of
neighbor-joining~\cite{NJ}, a well-established hierarchical clustering algorithm for inferring dendrograms (trees). 
We thus used \textsf{NINJA} to infer the respective tree $T_1$ under the neighbor-joining criterion.
We also inferred the tree $T_2$ by following the same pipeline, but by considering $\LW(\mina_{x},\mina_{y})$ as distance metric, 
as well as the tree $T_3$ by using the {\em basic} dataset as input of this pipeline and $\LW(\mina_{\tilde{x}},\mina_{\tilde{y}})$ as distance metric.
Hence, notice that $T_3$ represents the original tree. 
Finally, we computed the pairwise Robinson-Foulds (RF) distance~\cite{RFdistance} between: $T_1$ and $T_3$; and $T_2$ and $T_3$.

Let us define \textit{accuracy} as the difference between 1 and the relative pairwise RF distance.
We repeated this experiment by simulating different datasets $<\alpha,\beta,\gamma,\delta,\epsilon>$ and measured the corresponding accuracy.  
The results in Table~\ref{tab:accuracy} (see $T_1$ vs. $T_3$) suggest that by considering $\LW(\mina_{\tilde{x}},\mina_{\tilde{y}})$ we can always 
re-construct the original tree even if the sequences have first been arbitrarily rotated (Corollary~\ref{lem:circ}). 
This is not the case (see $T_2$ vs. $T_3$) if we consider $\LW(\mina_{x},\mina_{y})$. Notice that $100\%$ accuracy denotes a (relative) pairwise RF distance of 0.
\begin{table}[!t]
\begin{center}
\scalebox{0.7}{
\begin{tabular}{l|c|c} \hline
Dataset $<\alpha,\beta,\gamma,\delta,\epsilon>$	& $T_1$ vs. $T_3$ &  $T_2$ vs. $T_3$ \\ \hline
$<12,2500,0.05,0.06,0.04>$		& 100\%		& 100\%\\
$<12,2500,0.20,0.06,0.04>$		& 100\%		& 88,88\%\\
$<12,2500,0.35,0.06,0.04>$		& 100\%		& 100\%\\
$<25,2500,0.05,0.06,0.04>$		& 100\%		& 100\%\\
$<25,2500,0.20,0.06,0.04>$		& 100\%		& 100\%\\
$<25,2500,0.35,0.06,0.04>$		& 100\%		& 100\%\\
$<50,2500,0.05,0.06,0.04>$		& 100\%		& 97,87\%\\
$<50,2500,0.20,0.06,0.04>$		& 100\% 	& 97,87\%\\
$<50,2500,0.35,0.06,0.04>$		& 100\%		& 100\%\\ \hline
\end{tabular}
}
\end{center}
\caption{Accuracy measurements based on relative pairwise RF distance}
\label{tab:accuracy}
\end{table}

\section{Final Remarks}

In this article, complementary to measures that refer to 
the composition of sequences in terms of their constituent patterns, we considered sequence comparison using
minimal absent words, information about what does not occur in the sequences.
We presented the first linear-time and linear-space algorithm to compare two sequences by considering {\em all} their minimal absent words (Theorem~\ref{the:maw}).
In the process, we presented some results of combinatorial interest, and also extended the proposed techniques to circular sequences.
The power of minimal absent words is highlighted by the fact that they provide a tool for sequence comparison that is as efficient for circular as it is for linear 
sequences (Corollary~\ref{lem:circ} and Theorem~\ref{the:cmaw}); whereas, this is not the case, for instance, using the general edit distance model~\cite{Maes}.
Finally, a preliminary experimental study shows the potential of our theoretical findings.

Our immediate target is to consider the following {\em incremental} version of the same problem: given an appropriate encoding of a comparison between sequences $x$ and $y$, 
can one incrementally compute the answer for $x$ and $ay$, and the answer for $x$ and $ya$, efficiently, where $a$ is an additional letter?
Incremental sequence comparison, under the edit distance model, has already been considered in~\cite{Landau1998}.

In~\cite{WABI2015}, the authors considered a more powerful generalization of the $q$-gram distance (see~\cite{U92} for definition) to compare $x$ and $y$. 
This generalization comprises partitioning $x$ and $y$ in $\beta$ blocks each, as evenly as possible, computing the $q$-gram distance between the corresponding block pairs, 
and then summing up the distances computed blockwise to obtain the new measure. We are also planning to apply this generalization to the similarity measure studied here
and evaluate it using real and synthetic data.

\section*{Acknowledgements}

We warmly thank Alice Heliou for her inestimable code contribution and Antonio Restivo for useful discussions. 
Gabriele Fici's work was supported by the PRIN 2010/2011 project ``Automi e Linguaggi Formali: Aspetti 
Matematici e Applicativi'' of the Italian Ministry of Education (MIUR) and by the ``National Group for Algebraic and Geometric Structures, and their Applications'' (GNSAGA -- INdAM).
Robert Merca{\c s}'s work was supported by the P.R.I.M.E. programme of  DAAD co-funded by BMBF and EU's 7th Framework Programme (grant 605728).
Solon  P.  Pissis's  work was  supported  by  a  Research  Grant (\#RG130720) awarded by the Royal Society.

\bibliographystyle{splncs03}

\begin{thebibliography}{10}
\providecommand{\url}[1]{\texttt{#1}}
\providecommand{\urlprefix}{URL }

\bibitem{nullrly}
Acquisti, C., Poste, G., Curtiss, D., Kumar, S.: Nullomers: Really a matter of
  natural selection? PLoS ONE  2(10) (2007)

\bibitem{MAW}
Barton, C., Heliou, A., Mouchard, L., Pissis, S.P.: Linear-time computation of
  minimal absent words using suffix array. BMC Bioinformatics  15,  388 (2014)

\bibitem{PPAM2015}
Barton, C., Heliou, A., Mouchard, L., Pissis, S.P.: Parallelising the
  computation of minimal absent words. In: PPAM. LNCS (2015)

\bibitem{SEA2015}
Barton, C., Iliopoulos, C.S., Kundu, R., Pissis, S.P., Retha, A., Vayani, F.:
  Accurate and efficient methods to improve multiple circular sequence
  alignment. In: SEA, LNCS, vol. 9125, pp. 247--258 (2015)

\bibitem{BeMiReSc00}
B{\'e}al, M., Mignosi, F., Restivo, A., Sciortino, M.: Forbidden words in
  symbolic dynamics. Advances in Applied Mathematics  25(2),  163–--193
  (2000)

\bibitem{Belazzougui2013}
Belazzougui, D., Cunial, F., K{\"a}rkk{\"a}inen, J., M{\"a}kinen, V.: Versatile
  succinct representations of the bidirectional {B}urrows--{W}heeler transform.
  In: ESA, LNCS, vol. 8125, pp. 133--144 (2013)

\bibitem{Chairungsee2012109}
Chairungsee, S., Crochemore, M.: Using minimal absent words to build phylogeny.
  Theoretical Computer Science  450(0),  109--116 (2012)

\bibitem{CHL07}
Crochemore, M., Hancart, C., Lecroq, T.: Algorithms on Strings. Cambridge
  University Press, New York, NY, USA (2007)

\bibitem{Crochemore98automataand}
Crochemore, M., Mignosi, F., Restivo, A.: Automata and forbidden words.
  Information Processing Letters  67,  111--117 (1998)

\bibitem{Domazet-Loso:2009:EEP:1671627.1671629}
Domazet-Lo\v{s}o, M., Haubold, B.: Efficient estimation of pairwise distances
  between genomes. Bioinformatics  25(24),  3221--3227 (2009)

\bibitem{fici}
Fici, G.: Minimal Forbidden Words and Applications. Ph.D. thesis,
  Universit\'{e} de Marne-la-Vall\'{e}e (2006)

\bibitem{indLCP}
Fischer, J.: Inducing the {LCP}-array. In: WADS, LNCS, vol. 6844, pp. 374--385
  (2011)

\bibitem{Fischer11}
Fischer, J., Heun, V.: Space-efficient preprocessing schemes for range minimum
  queries on static arrays. SIAM Journal of Computing  40(2),  465--492 (2011)

\bibitem{indelible}
Fletcher, W., Yang, Z.: {INDELible}: A flexible simulator of biological
  sequence evolution. Molecular Biology and Evolution  26(8),  1879--1888
  (2009)

\bibitem{DBLP:conf/isit/FukaeOM12}
Fukae, H., Ota, T., Morita, H.: On fast and memory-efficient construction of an
  antidictionary array. In: ISIT. pp. 1092--1096. {IEEE} (2012)

\bibitem{minabpro}
Garcia, S.P., Pinho, O.J., Rodrigues, J.M.O.S., Bastos, C.A.C., G, P.J.S.:
  Minimal absent words in prokaryotic and eukaryotic genomes. PLoS ONE  6
  (2011)

\bibitem{mtDNA-phylo}
Goios, A., Pereira, L., Bogue, M., Macaulay, V., Amorim, A.: {mtDNA} phylogeny
  and evolution of laboratory mouse strains. Genome Research  17(3),  293--298
  (2007)

\bibitem{WABI2015}
Grossi, R., Iliopoulos, C.S., Merca{\c s}, R., Pisanti, N., Pissis, S.P.,
  Retha, A., Vayani, F.: Circular sequence comparison with $q$-grams. In: WABI,
  LNCS, vol. 9289, pp. 203--216 (2015)

\bibitem{LCE}
Ilie, L., Navarro, G., Tinta, L.: The longest common extension problem
  revisited and applications to approximate string searching. Journal of
  Discrete Algorithms  8(4),  418--428 (2010)

\bibitem{Landau1998}
Landau, G.M., Myers, E.W., Schmidt, J.P.: Incremental string comparison. SIAM
  J. Comput.  27(2),  557--582 (1998)

\bibitem{Maes}
Maes, M.: On a cyclic string-to-string correction problem. Information
  Processing Letters  35(2),  73--78 (1990)

\bibitem{SA}
Manber, U., Myers, E.W.: Suffix arrays: A new method for on-line string
  searches. SIAM Journal of Computing  22(5),  935--948 (1993)

\bibitem{Mignosi02}
Mignosi, F., Restivo, A., Sciortino, M.: Words and forbidden factors.
  Theoretical Computer Science  273(1-2),  99--117 (2002)

\bibitem{conf/gcb/MosigHS06}
Mosig, A., Hofacker, I.L., Stadler, P.F.: {Comparative Analysis of Cyclic
  Sequences: Viroids and other Small Circular RNAs}. In: GCB, LNI, vol.~83, pp.
  93--102 (2006)

\bibitem{Nong:2009:LSA:1545013.1545570}
Nong, G., Zhang, S., Chan, W.H.: Linear suffix array construction by almost
  pure induced-sorting. In: DCC. pp. 193--202. IEEE (2009)

\bibitem{6979851}
Ota, T., Morita, H.: On a universal antidictionary coding for stationary
  ergodic sources with finite alphabet. In: ISITA. pp. 294--298. {IEEE} (2014)

\bibitem{DBLP:conf/isit/OtaM13}
Ota, T., Morita, H.: On antidictionary coding based on compacted substring
  automaton. In: ISIT. pp. 1754--1758. {IEEE} (2013)

\bibitem{Pinho2009}
Pinho, A.J., Ferreira, P.J.S.G., Garcia, S.P.: {On finding minimal absent
  words}. BMC Bioinformatics  11 (2009)

\bibitem{RFdistance}
Robinson, D., Fould, L.: Comparison of phylogenetic trees. Mathematical
  Biosciences  53(1-2),  131--147 (1981)

\bibitem{NJ}
Saitou, N., Nei, M.: The neighbor-joining method: a new method for
  reconstructing phylogenetic trees. Molecular Biology and Evolution  4(4),
  406--425 (1987)

\bibitem{Silva02042015}
Silva, R.M., Pratas, D., Castro, L., Pinho, A.J., Ferreira, P.J.S.G.: Three
  minimal sequences found in {Ebola} virus genomes and absent from human {DNA}.
  Bioinformatics  (2015)

\bibitem{U92}
Ukkonen, E.: Approximate string-matching with $q$-grams and maximal matches.
  Theoretical Computer Science  92(1),  191--211 (1992)

\bibitem{ninja}
Wheeler, T.J.: Large-scale neighbor-joining with {NINJA}. In: WABI, LNCS, vol.
  5724, pp. 375--389 (2009)

\end{thebibliography}

\end{document}